\documentclass[12pt]{article}
\usepackage{amssymb,amsmath,amsthm}
\usepackage{amscd}
\usepackage{mathrsfs}
\usepackage[english]{babel}
\usepackage[utf8]{inputenc}
\usepackage{graphicx}
\usepackage{graphics}
\usepackage{authblk}
\usepackage{hyperref}
\usepackage[section]{placeins}
\newtheorem{theorem}{Theorem}

\newtheorem{remark}{Remark}

\newtheorem{prop}{Proposition}[section]
\theoremstyle{definition}

\title{Efficient semiclassical approximation for bound states in graphene in magnetic field with a~small trigonal warping correction}
\newcommand*{\email}[1]{%
    \normalsize\href{mailto:#1}{#1}\par
    }
\author{V.~V.~Rykhlov}
\affil{Lomonosov Moscow State University\\ \email{vladderq@gmail.com}}
\date{}
\begin{document}
\maketitle
\section*{Abstract}
This paper is devoted to the construction of semiclassical spectrum and efficient (simple to implement) explicit semiclassical asymptotic eigenfunctions of the Dirac operator for relatively high-energy bound states in graphene in magnetic field, considering the effect of trigonal warping~\cite{Kats, Koen} to be small. It turns out that the asymptotic spectrum of the operator remains unchanged under such a~perturbation due to the symmetry of the problem rather than the smallness of this correction. 

However, the behavior of asymptotic eigenfunctions is quite different; they are significantly affected by trigonal warping that leads to the breaking of certain symmetries. Density plots of asymptotic eigenfunctions can indicate what might be observed using a~scanning tunneling microscope. Our approach to constructing asymptotic solutions is based on developments of works~\cite{DNAiry, Dioph, Constr}, which present a~new method for constructing the solution, simplifying practical application.

\section{Introduction}
Graphene, a~single layer of carbon atoms arranged in a~hexagonal lattice, has garnered significant attention due to its exceptional electronic properties and potential applications in various fields. Recent studies, particularly those by Katsnelson~\cite{Kats}, have advanced the understanding of graphene's unique characteristics. One of phenomena appearing here is trigonal warping. This distortion of the electron band structure near the Dirac points leads to anisotropic behavior of charge carriers and impacts electronic transport properties.

Our aim is to construct the efficient representations (\textit{efficiency} here means that it must be simple and quick to compute and plot the resulting eigenfunctions using, e.g., Wolfram Mathematica or Maple) for \textit{formal asymptotic eigenfunctions} (for a~rigorous definition, see below) of the Dirac operator for graphene with trigonal warping correction, see Eq.~\eqref{main_eq_TW}. The square of the absolute value of the eigenfunction is proportional to the local density of states (LDOS)~\cite{Tersoff, Ukr} and can be measured using a~scanning tunneling microscope (STM), see also~\cite{BrDobrKatsMin}. Thus, it is expected that for the class of problems under consideration, we can predict what will be observed with STM.

Let us proceed to problem setting. Tight-binding approximation of dynamics in graphene is given by the eigenequation for 2-D~Dirac operator with a~symbol in the form of self-adjoint matrix (see \cite{Koen}):
\begin{equation}
\begin{gathered}
\widehat{H}_{D} \Psi=\mathscr E\Psi,\qquad  \widehat{H}_{D}=H_{D}(\hat k,y),\qquad
H_{D} = \hbar v_F \begin{pmatrix} 0 & k_1 -i k_2 \\ k_1+ik_2 & 0 \end{pmatrix} + m(y) \sigma_3 + u(y) \mathbf I,
\end{gathered}
\end{equation}
where $\hbar$ is the Planck constant, $k = \begin{pmatrix} k_1 \\ k_2 \end{pmatrix}$ is the wave vector and $\hat k_j = -i \partial / \partial y_j$, $v_F$ is Fermi velocity in graphene, $\sigma_3 = \begin{pmatrix} 1 & 0 \\ 0 & -1\end{pmatrix}$ is Pauli matrix and $\mathbf I$ is the identity matrix.
Consider the \textit{trigonal warping} correction (see~\cite{Kats}, \cite{Koen}) of this operator: 
\begin{equation}
H_{TW} = H_{D} - \frac{3}{8}t a_{CC}^2 \begin{pmatrix} 0 & (k_1+i k_2)^2 \\ (k_1- i k_2)^2 & 0 \end{pmatrix},
\end{equation}
where $t = 3eV$, $a_{CC} = 0.142$ nm.
Using the relationship
$\hbar v_F = \frac{3}{2}t a_{CC}$
for parameters $\hbar$, $v_F = 0.97 \cdot 10^6$ m/s, $t$ and $a_{CC}$ nm, introducing a~magnetic field term, and denoting a~typical energy by~$E_0$ and a~typical length scale by~$l$,
we are then able to transform the matrix-valued symbol to
\begin{equation}\label{TW_symbol}
H_{TW} = E_0 \mathcal L(p, x) - \frac{E_0^2}{6t} \begin{pmatrix} 0 & (\textbf{p}_1+i\textbf{p}_2)^2 \\ (\textbf{p}_1-i\textbf{p}_2)^2 & 0 \end{pmatrix},
\end{equation}
here $\mathbf{p}_j = p_j + A_j(x)$, $A_1 = \frac{Bx_2}{2}$, $A_2 = - \frac{B x_1}{2}$ ($B = \frac{e v_F}{E_0}\mathbf B l$, where $e = 1.602 \cdot 10^{-19}$  is the electron charge, $\mathbf B$ is the value of magnetic flux density), $x_j = y_j / l$, $p_j$ is a~symbol of $-ih\partial / \partial x_j$, where $h = \frac{\hbar v_F}{E_0 l}$ is a~dimensionless small parameter, and
$$
\begin{gathered}
\mathcal L(p, x) = \begin{pmatrix} U(x,h)+M(x,h)& {\mathbf{p}}_1 - i{\mathbf{p}}_2 \\
 {\mathbf{p}}_1+i{\mathbf{p}}_2& U(x,h)-M(x,h) \end{pmatrix},\quad  M(x,h) = m(x,h)/E_0, \quad U(x,h) = u(x,h)/E_0.
\end{gathered}
$$

We are constructing asymptotic \textit{solutions} of the eigenequation
\begin{equation}\label{main_eq_TW}
\widehat{\mathcal L}_{TW} \Psi = \mathscr E \Psi, \qquad \mathcal L_{TW} = \mathcal L(p, x) + \mu \begin{pmatrix} 0 & (\textbf{p}_1+i\textbf{p}_2)^2 \\ (\textbf{p}_1-i\textbf{p}_2)^2 & 0 \end{pmatrix},
\end{equation}
where $\mu \equiv h \gamma = \frac{E_0}{6t}$, and since $h$ is small, the symbol ${\mathcal L}_{TW}$ is a~small perturbation of $\mathcal L$.
It is necessary to clarify what is meant by the \textit{solution} of this equation. A~pair $(\mathscr E, \Psi(x,h))$ is said to be a~\textit{solution} of Eq.~\eqref{main_eq_TW}, if $\|(\widehat{\mathcal L}_{TW}-\mathscr E) \Psi(x,h)\|_{L^2} = O(h^{1+\delta})$, $\delta>0$. For the principal symbol to be integrable, we assume $U(x,h)$ to be radially symmetric and $M$ to be either radially symmetric or small (i.e. $M = \sqrt{h} \widetilde M$).

\begin{remark}
The mapping $f\to \widehat f$ takes the symbol $f = f(p,x)$ to the $h$-pseudodifferential operator $\widehat f = f(\overset{1}{-ih\partial/\partial x}, \overset{2} x)$, where the numbers $1$ and $2$ over the arguments indicate the order of action of the respective operators, see~\cite{Feyn}.
\end{remark}
\begin{remark}
One may consider $\mu$ as a~parameter of order $h^{1/2+\delta}$ for $\delta>0$ or even as a~parameter independent of $h$. In general, this would lead to the destruction of integrability of the principal symbol $\mathcal L(p,x)$, since the trigonal warping correction in that case should be considered as a~part of it, so the question about the behavior of the asymptotic functions in this case remains open. Nevertheless, some calculations can be done with considering $\mu$ as a~free parameter, and we will do so whenever possible.
\end{remark}

\begin{remark}[on the smoothness of the eigenfunctions]
Let us consider $\widehat {\mathcal L}_{TW}$ as a~classical pseudodifferential operator. The principal symbol (in terms of order of differentiation instead of order with respect to the power of $h$) of operator~\eqref{main_eq_TW} is given by
\begin{equation}\label{class_prin_sym}
\mathcal L^{(2)} = \mu h^2 \begin{pmatrix} 0 & (\xi_1 + i \xi_2)^2 \\ (\xi_1 - i \xi_2)^2 & 0  \end{pmatrix},
\end{equation}
where $\xi_j$ stands for the symbol of $-i \partial / \partial x_j$.
Since the symbol \eqref{class_prin_sym} is 
nondegenerate (here $\operatorname{det} \mathcal L^{(2)} = - \mu h^2 |\xi|^4$) for any $\xi \in  \mathbb S^2$, any fixed $h>0$ and $\mu \ne 0$, then
$\widehat {\mathcal L}_{TW}$ is elliptic. In other words, it is nonuniformly elliptic with respect to~$h$ and~$\mu$ in any neighborhood of zero.
Nevertheless, $h$~and~$\mu$ can be considered constant for any (exact) solution~$\Psi$ of the eigenequation (i.e. $\widehat {\mathcal L}_{TW} \Psi = \mathscr E \Psi$), thus, $\Psi \in C^\infty$. It is important to note that our method allows us to construct \textit{formal asymptotic eigenfunctions} \cite{Constr}, rather than asymptotics for \textit{exact} eigenfunctions. However, their belonging to the same class of functions gives hope that one day it will be possible to prove that our asymptotic eigenfunctions approximate the exact ones. E.g., in~\cite{AnDobrTsv} it was proved for one-dimensional Schr\"odinger operator with simple spectrum.
\end{remark}

\section{Scalarization}
As in previous paper~\cite{Constr}, we reduce the equation for the matrix-valued operator to a~pair of scalar problems. In order to do so, we find matrix-valued symbols $\chi=\chi (p,x,h)$ and $\mathbb H=\mathbb H(p,x,h)$ such that
$$
\chi = \chi_0+h\chi_1+O(h^2),\quad \mathbb H = \begin{pmatrix} \mathbb H^+ & 0 \\ 0 & \mathbb H^-\end{pmatrix},\quad \widehat {\mathcal L}_{TW}(p,x) \widehat\chi = \widehat\chi \widehat{\mathbb H} + O(h^2),
$$
where the last relation means that the symbols of the operators on the left- and right-hand sides differ by $O(h^2)$.
The principal symbols of $\widehat{\mathbb H}^\pm$ (Hamiltonians for electron and hole bands~\cite{Koen}) are eigenvalues of the matrix-valued symbol $\mathcal L(p,x)$:
$$
\mathcal L(p, x) \chi_0^\pm = \mathbb H_0^{\pm} \chi_0^\pm, \qquad  \mathbb H_0^{\pm} = U - E \pm \sqrt{M^2+\mathbf{p}^2},\qquad
\chi_0 = \begin{pmatrix} \chi_{0,1}^+ & \chi_{0,1}^-\\ \chi_{0,2}^+ & \chi_{0,2}^-\end{pmatrix},
$$
where~$\chi$ is the matrix whose columns are the corresponding eigenvectors, moreover,
$$
\begin{gathered}
\|\chi_0^\pm\| \equiv \sqrt{|\chi_{0,1}^\pm|^2+|\chi_{0,2}^\pm|^2} = 1, \quad
\chi_0^\pm = \frac{1}{\sqrt{2 (M^2 + \mathbf{p}^2 \pm M \sqrt{M^2 + \mathbf{p}^2})}} \begin{pmatrix} \mathbf{p}_1 - i \mathbf{p}_2 \\ -M \mp \sqrt{M^2 + \mathbf{p}^2} \end{pmatrix}.
\end{gathered}
$$
Note that the multiplicity of these eigenvalues is identically equal to~$1$ in the domain where $M^2+\mathbf{p}^2 \ne 0$. At the points where $M^2+\mathbf{p}^2 = 0$, the multiplicity changes and that leads to many difficulties. In what follows, we impose natural conditions so that the multiplicity remains constant.

Further, we use the well-known formula for symbol of the first correction (see \cite{4authorsAiry}):
$$
\mathbb H_{1, TW}^{\pm}= \Big\langle\overline{\chi^{\pm}_0}, \mathcal L_1 {\chi^{\pm}_0}\Big\rangle  -i\Big\langle\overline{\chi^{\pm}_0},\frac{d}{dt}\chi^{\pm}_0 \Big\rangle-i\Big\langle\overline{\chi^{\pm}_0},\sum_{j=1}^2\Big(\frac{\partial \mathcal{L}}{\partial p_j}-\mathbf{I}\frac{\partial \mathbb{H}_{\pm}^0}{\partial p_j}\Big)\frac{\partial \chi^{\pm}_0}{\partial x_j}\Big\rangle,
$$
where $\langle \cdot, \cdot \rangle$ is the standard dot product,  $\mathbf I$ is the identity matrix,  $\mathcal L_1$ is the first correction to the matrix-valued symbol (i.e., $\mathcal L_1 = \gamma \begin{pmatrix} 0 & (\textbf{p}_1+i\textbf{p}_2)^2 \\ (\textbf{p}_1-i\textbf{p}_2)^2 & 0 \end{pmatrix}$),
$\frac{d}{dt}$ is the derivative along the Hamiltonian vector field ${\mathit v}_{\mathbb H_0^\pm}$.
It is easy to see, that the scalar correction formulas for $\mathcal L_{TW}$ and $\mathcal L$ differ only by term depending on $\mathcal L_1$.
Thus, 
$$
\mathbb H_{1, TW}^\pm =\mathbb H_{1}^\pm - \gamma \frac{(M \pm \sqrt{M^2 + \mathbf{p}^2})(\mathbf{p}_1^3 - 3\mathbf{p}_1 \mathbf{p}_2^2)}{M^2 + \mathbf{p}^2 \pm M\sqrt{M^2 +\mathbf{p}^2}},
$$
where (see \cite{Constr})
$$
\mathbb{H}_1^\pm = \pm \frac{B}{2\sqrt{\mathbf{p}^2 + M^2}} +
\frac{(\mathbf{p}_\perp, \nabla(U+M))}{2\sqrt{\mathbf{p}^2+M^2}(\sqrt{\mathbf{p}^2+M^2} \mp M)} -
\frac{i}{2}{\rm tr}\,\frac{\partial^2 \mathbb{H}_\pm^0}{\partial p \partial x}, \qquad \mathbf{p}_\perp = \begin{pmatrix} - \mathbf{p}_2 \\ \mathbf{p}_1 \end{pmatrix}.
$$

\begin{remark}
Considering $\mu$ as a~free parameter, we get the following Hamiltonians
\begin{equation}
\mathbb{H}_{0, TW}^{\pm} = U \pm \sqrt{M^2 + |\textbf{p}|^2 + 2 \mu(\textbf{p}^3-3\textbf{p}_1\textbf{p}_2^2)+\mu^2 |\textbf{p}|^4},
\end{equation}
and eigenvectors that correspond to them are as follows
\begin{equation}
\chi_{0, TW}^{\pm} = \frac{1}{\sqrt{2 (\xi \mp M \sqrt{\xi} )}} \begin{pmatrix} \textbf{p}_1 - \textbf{p}_2 +\mu (\textbf{p}_1+i\textbf{p}_2)^2  \\ -M \pm \sqrt{\xi}\end{pmatrix}
,\quad \xi = M^2+|\textbf{p}|^2 + 2 \mu (\textbf{p}_1^3-3\textbf{p}_1\textbf{p}_2^2)+\mu^2 |\textbf{p}|^4.
\end{equation}
Under the assumption that $\mu = O(h^\beta)$, $0<\beta < 1$, the correction $\mathbb{H}_1^\pm$ (see below) remains the same as in the unperturbed case, as the trigonal warping then is included into the Hamiltonian.
\end{remark}

\section{Polar Coordinates and Main Result}
Further, using the semiclassical analog of Maupertuis--Jacobi principle (see~\cite{Constr}, and also~\cite{AniDobrKlevTir, DobrMinRo}), we bring our problem to the following:
\begin{equation}\label{HTW_Cartesian}
\mathcal H_{TW}(\overset{1}{\hat p}, \overset{2}{x}, E, \lambda, h) \psi = 0, \quad \mathcal H_{TW}(p,x,E,\lambda, h) = \mathcal H_0(p,x,E) + h \mathcal H_{1, TW}^\pm(p,x,E, \lambda),
\end{equation}
$$
\mathcal H_0(p,x, E) = (U(x)-E)^2 - (\mathbf{p}^2+M^2), \qquad
\mathcal H_{1, TW}^\pm =  \mathcal H_{1}^\pm+ 2 \gamma (\mathbf{p}_1^3 - 3 \mathbf{p}_1 \mathbf{p}_2^2),
$$
$$
\mathcal{H}_1^\pm =  - B \pm \frac{\langle \mathbf{p}_\perp, \nabla_x (U+M) \rangle}{\sqrt{M^2 + \mathbf{p}^2} \mp M}
- 2 (U-E) \lambda
\mp i  \frac{\langle \mathbf{p}, \nabla_x U \rangle}{\sqrt{M^2 + \mathbf{p}^2}}.
$$
\begin{remark}
As seen from formula for $\mathcal H_{1, TW}^\pm$, considering
trigonal warping as a~small perturbation of the Dirac equation, we must just
add a~term $2 \gamma (\mathbf{p}_1^3 - 3 \mathbf{p}_1 \mathbf{p}_2^2)$
to the correction, which leads to a~crucial change in the solution of the transport equation (the amplitude).
\end{remark}

The standard approach \cite{Masl, CdV, LazKAM} of constructing spectral series for an~operator with an~integrable principal symbol is to consider a~family of invariant Liouville tori of the Hamiltonian system (in our two-dimensional case, this family is two-parametric) and take a~discrete subfamily consisting of tori~$\Lambda$ satisfying the Bohr--Sommerfeld--Maslov quantization condition
$$
\frac{1}{2\pi} \oint_{\gamma_j} p\, dx = h\Big(k_j + \frac{m_j}{4}\Big),\quad k_j\in\mathbb Z,\quad j = 1,2,\dots,\operatorname{dim}\Lambda,
$$
for all basis cycles $\gamma_j\subset\Lambda$, where $m_j$ are the Maslov indices of these cycles. In that case, the Maslov canonical operator $K_\Lambda$  (see~\cite{Masl}; for simplicity, it can be interpreted as a~generalization of the WKB method), which takes an~\textit{amplitude} $A\in C^\infty(\Lambda,\mathbb C)$ to a~smooth rapidly oscillating function $[K_\Lambda A](x,h)$, is well-defined. If, moreover, $\lambda \in\mathbb C$, $A \in C^\infty(\Lambda,\mathbb C)$ is a~solution of the \textit{transport equation} (see~\cite{Masl}), then $K_\Lambda A$ is an~asymptotic eigenfunction of the given operator corresponding to an~eigenvalue $\mathscr E = E+h\lambda$.

The following theorem demonstrates how to reconstruct a~solution of the vector problem from the solution of the scalar one, provided the latter is represented by the canonical operator.

\begin{prop}[\cite{Constr}, Proposition~2]
Let Lagrangian torus $\Lambda \subset \mathbb{R}^4_{(p,x)}$ and parameter~$E$ be chosen such that quantization conditions (see \cite{Constr}) are met, and
$$
\mathcal H_0 (p,x,E) \big|_\Lambda = 0, \quad \Lambda \subset \{(p,x) \big| \pm (U(x)-E)>\delta>0 \}.
$$
Let $A \in C^\infty(\Lambda, \mathbb{C})$, $\psi_\pm = K_\Lambda A$ and parameters
$E$, $\lambda$ be such that $\mathcal H_{TW}^\pm (\overset{1}{\hat p}, \overset{2}{x}, E, \lambda, h) \psi_\pm = \mathscr O(h^{1+\alpha})$ for
some $\alpha >0$. Then $\mathscr E = E+h\lambda$, $\Psi_\pm = \widehat \chi_\pm \psi_\pm$ is an~asymptotic solution of \eqref{main_eq_TW}, i.e.
$$
\Big( \widehat{\mathcal L}_{TW} - (E+h\lambda) \Big) \Psi_\pm = \mathscr O(h^{1+\alpha}).
$$
\end{prop}
\begin{remark}
$\chi^\pm$ can be expanded into series $\chi_0^\pm + h\chi_1^\pm+O(h^2)$. We do not provide the formula for~$\chi_1^\pm$ due to its substantial size.
Considering only $\chi_0^\pm$ instead of $\chi^\pm$, we obtain a~leading-order term of an~asymptotic solution.
\end{remark}
\begin{remark}
The residual, denoted here by $\mathscr O(h^{1+\alpha})$, is understood in the sense of the space $\mathbf{H}^{s,h}_{\rm loc}$, where $\| f(x,h)\| = \int |\widetilde f| (1+|p|^2)^{s/2}$ and $\widetilde f$ is the $h$-Fourier transform of $f$, see~\cite{Constr}.
\end{remark}
\begin{proof}
This proposition has already been proved in \cite{Constr} for $\mathcal H_0$ and $\mathcal H_1^\pm$.
To complete the proof, it is necessary to show that $(\mathbb{H}_0^\mp - E) \big (- \gamma \frac{(M \mp \sqrt{M^2 + \mathbf{p}^2})(\mathbf{p}_1^3 - 3\mathbf{p}_1 \mathbf{p}_2^2)}{M^2 + \mathbf{p}^2 \pm M\sqrt{M^2 +\mathbf{p}^2}} \big) \big|_\Lambda =
2\gamma (\mathbf{p}_1^3 - 3 \mathbf{p}_1 \mathbf{p}_2^2)$. It follows immediately from the identities $\mathbb{H}_0^\mp \big|_\Lambda = 2(U - E)$ and $\pm \sqrt{M^2 +\mathbf{p}^2} \big|_\Lambda  = E - U$.
\end{proof}

Recalculating the symbol $\mathcal H_{TW}$ (see Eq.~\eqref{HTW_Cartesian}) to the polar coordinates by the formulas
$$
(x_1,x_2) = r(\cos{\varphi}, \sin{\varphi}), \quad
p_1 = p_r\cos{\varphi}-\frac{p_\varphi\sin{\varphi}}{r},\quad
p_2 = p_r\sin{\varphi}+\frac{p_\varphi\cos{\varphi}}{r}
$$
and changing the sign for convenience (we are seeking the asymptotic eigenfunctions, so the omitted sign does not affect anything),
for the small $M = \sqrt{h} \widetilde M$ we get (with arguments of functions omitted)
$$
\begin{gathered}
H_0 = p_r^2+\mathbf{R}_\varphi^2- (U-E)^2,\\
H_{1, TW}^\pm = B -\frac{\mathbf{R}_\varphi \partial U/\partial r}{U-E} + \widetilde M^2 + 2(U-E)\lambda -2\gamma \Big(p_r (p_r^2 - 3\mathbf{R}_\varphi^2)\cos{3\varphi} + \mathbf{R}_\varphi(\mathbf{R}_\varphi^2-3 p_r^2)\sin{3\varphi}\Big) + i p_r W^+,
\end{gathered}
$$
where $\mathbf{R}_\varphi = \frac{p_\varphi}{r}-\frac{Br}{2}$ and $W^+=\frac{\partial U/\partial r}{U-E}$ (note that there is a~little inaccuracy in~\cite{Constr}, where ``$\mp$'' instead of ``$-$'' appears), and
 for the radially symmetric $M=M(r)$ we get 
$$
\begin{gathered}
H_0 = p_r^2+\mathbf{R}_\varphi^2 + M^2 - (U-E)^2,\\
H_{1, TW}^\pm = B - \frac{\mathbf{R}_\varphi \partial (U+M)/\partial r}{U-E+ M} + 2(U-E)\lambda - 2\gamma \Big(p_r (p_r^2 - 3\mathbf{R}_\varphi^2)\cos{3\varphi} + \mathbf{R}_\varphi(\mathbf{R}_\varphi^2-3 p_r^2)\sin{3\varphi}\Big) + ip_r W^\pm,
\end{gathered}
$$
where $W^- = \frac{\partial U/\partial r}{U-E} + 2 M \frac{\partial M/\partial r}{(U-E)^2}$.

When changing coordinates from~$x$ to $y$, the solution expressed as a~canonical operator in $y$ must be multiplied by $\sqrt{\operatorname{det} \frac{\partial y}{\partial x}}$ to yield the solution in the original $x$ coordinates, see~\cite{OpMet}. It is taken into account in formula~\eqref{leading_term_of_vector_sol}.

For completeness, let us briefly outline the results from~\cite{Constr}.
Let~$r = \mathcal R(t)$, $p_r = \mathcal P(t)$ be a~$T$-periodic solution of dynamical equations on~$\Lambda$
$$
\dot r = \frac{\partial H_0}{\partial p_r},\quad \dot p_r = -\frac{\partial H_0}{\partial r},
$$
$r_-$, $r_+$ be the endpoints of the segment, which is a~projection of the curve $(\mathcal R(t), \mathcal P(t))$ to $r$-axis,
and $\theta_1 = 2\pi t/T\ \operatorname{mod} 2\pi$. Denote
$$
\begin{gathered}
U_{\rm eff}(r) = \mathbf{R}_\varphi(r,p_\varphi)^2 + M(r)^2 - (U(r)-E)^2 \operatorname{mod} O(h),\quad
\mathbf{\Phi}(r_1,r_2) = \int_{r_1}^{r_2} \sqrt{-U_{\rm eff}(\tilde r)}\, d\tilde r,\\ \mathbf{\Phi}^{\rm out}(r) = \mathbf{\Phi}(r,r_+),\quad \mathbf{\Phi}^{\rm in}(r) = \mathbf{\Phi}(r_-,r),\quad
\beta^{\rm in} = \mathbf{\Phi}(r_-, r_+),\quad \beta^{\rm out}=0,\\ Q(\theta_1) = \int_0^{\theta_1} \Big(\frac{1}{R^2(\widetilde\theta)} - \Omega\Big)\, d\widetilde\theta, \quad
\Omega = \frac{1}{2\pi}\int_0^{2\pi} \frac{d\widetilde\theta}{R^2(\widetilde\theta)}, \quad \omega_1 = \frac{2\pi}{T},
\end{gathered}
$$
here $\operatorname{mod} O(h)$ removes the mass if it is small, ``$\rm in/out$'' means the caustic from which paths with respective phases are issued ($r=r_-$ and $r=r_+$ respectively). Denote $D_{\rm in} \equiv \{ (r,\varphi) | r\in [r_-, r_+-\delta]\}$, $D_{\rm out} \equiv \{ (r,\varphi) | r\in [r_- + \delta, r_+]\}$ for some $\delta>0$.
Let also $\theta_2 = \varphi - \frac{2p_\varphi}{\omega_1}Q(\theta_1)$, $\omega_2=-B+2p_\varphi \Omega$; let
$\Theta^\pm(r,\varphi)$ be a~pair of inverse mappings of the projection $\Lambda \to \mathbb R^2_{r,\varphi}$ to the coordinate space ("$+$" for $p_r>0$ and "$-$" for $p_r<0$),  and $\theta^0 \in\Lambda$ be an~initial point, see~\cite{Masl}, corresponding to $(r_0,\varphi_0)$. Dynamics on $\Lambda$ admits an~invariant measure(volume form) $d\mu = d\theta_1\wedge d\theta_2$. Introduce
$$
\mathscr A_\pm(r,\varphi) = A\Big(\Theta^\pm_1(r), \varphi - \frac{2p_\varphi}{\omega_1} Q(\Theta^\pm_1(r))\Big),\quad \mathscr A_{\rm ev/odd}(r,\varphi) = \frac{\mathscr A_+(r,\varphi)\pm \mathscr A_-(r,\varphi)}{2}.
$$

The standard semiclassical theory allows one to construct the solution of the eigenequation~\eqref{main_eq_TW} if and only if quantization conditions are met (see~\cite{Masl}). These conditions can be in contradiction with the Diophantine condition for the frequencies $\omega_1$, $\omega_2$, which is necessary for the transport equation to be solvable \cite{Dioph}. Our approach provides a~way to construct a~solution even if quantization conditions are violated, more precisely, in case the Diophantine torus $\overline\Lambda$ lies in $O(h)$-neighborhood of the Bohr--Sommerfeld torus $\Lambda$ \cite{Dioph}, \cite{Constr}.

For our problem, action variables, which are canonically conjugate to $\theta_1$, $\theta_2$, are given by
$$
I_1 = \frac{1}{\pi}\int_{r_-}^{r_+} \sqrt{-U_{\rm eff}(r)} dr,\quad I_2 = p_\varphi,
$$
and quantization conditions are as follows
$$
I_1 = h \Big(\nu_1+\frac{1}{2}\Big),\quad I_2 = h\nu_2,
$$
see~\cite{Constr}. Define the \textit{action defect} (see~\cite{Constr}, and also~\cite{Laz}) by $
q_\nu \equiv (q_1, q_2) = \Big(h \Big(\nu_1+\frac{1}{2}\Big) - I_1, h\nu_2 - I_2\Big)$, 
where~$I_1$,~$I_2$ are the actions of $\overline\Lambda$. According to~\cite[Theorem~1]{Dioph}, the canonical operator $K_{\overline\Lambda}[\overline A e^{i \langle q, \theta\rangle /h}]$ is well-defined. Now we are ready to proceed to the following theorem.

\begin{theorem}[\cite{Constr}, Theorem~4, with some redesignation]\label{TheoremConstr}
Let $\lambda \in\mathbb R$ be such that
\begin{equation}\label{lambda_condition}
\iint_{\overline\Lambda} \operatorname{Re} H_{1,TW}(p,x,E,\lambda) \big|_{\overline \Lambda}\, d\theta_1 d\theta_2 = 0,
\end{equation}
the vector of frequencies $(\omega_1,\omega_2)$ of $\overline \Lambda$ be Diophantine, $q_\nu = O(h)$, 
$a_\pm(\theta)$ be a~solution of the reduced transport equation
$$
-i \frac{da_\pm}{dt} + \operatorname{Re}\mathcal H_{1, TW}^\pm\big|_{\overline\Lambda} a_\pm = 0,
$$
and the condition
$$
\overline\Lambda\subset \{(r,\varphi,p_r,p_\varphi)\big| \pm (U(r)-E)>0\}
$$
be met.
Set 
$$
A_\nu^\pm =  \frac{1}{2} \int W^\pm\, dr e^{ a_\pm(\theta) + \frac{i}{h}\langle q, \theta\rangle}.
$$ 
Then the pair $(\mathscr E, \Psi_\pm)$, where
$$
\mathscr E \equiv E_\nu + h\lambda,\quad \Psi_\pm(x) := \widehat \chi^\pm \big( \sqrt{r} K_{\overline\Lambda} A_\nu^\pm \big|_{(r,\varphi)\to (x_1,x_2)}\big),
$$
is a~solution of problem~\eqref{main_eq_TW}. Namely, $\|\widehat{\mathcal L}\Psi_\pm - \mathscr E\Psi_\pm\|_{L^2} = O(h^{3/2})$, where $K_{\overline\Lambda}\equiv K_\Lambda$ admits the representation
\begin{equation}\label{main_repr}
\begin{gathered}
[K_{\Lambda} A_\nu^\pm]^{\rm in/out}(r,\varphi) \asymp \sqrt{2\pi\omega_1} \frac{e^{\frac{i}{h} \big(p_\varphi(\varphi-\varphi_0) - \beta^{\rm in/out} +\mathbf \Phi(r_0,r_+)\big)}}{|U_{\rm eff}(r)|^{1/4}}
\Big[\sigma^{\rm in/out}_1 \mathscr A_{\rm ev}(r,\varphi) \Big(\frac{3\mathbf \Phi^{\rm in/out}(r)}{2h}\Big)^{1/6}\\ \times {\rm Ai}\Big(-\Big(\frac{3\mathbf \Phi^{\rm in/out}(r)}{2h}\Big)^{2/3}\Big) 
+ \sigma^{\rm in/out}_2 \mathscr A_{\rm odd}(r,\varphi) \Big(\frac{3\mathbf \Phi^{\rm in/out}(r)}{2h}\Big)^{-1/6} {\rm Ai}'\Big(-\Big(\frac{3\mathbf \Phi^{\rm in/out}(r)}{2h}\Big)^{2/3}\Big)\Big],
\end{gathered}
\end{equation}
where $\sigma_1^{\rm out} = \sigma_2^{\rm in} = e^{\pi i/4}$ and $\sigma_2^{\rm out}=\sigma_1^{\rm in}  = e^{\pi i/4}$, in the domain $D_{\rm in/out}$. Here $\asymp$ means equality modulo $O(h)$ terms in the amplitude $A(\theta)$, and the vector-valued function
\begin{equation}\label{leading_term_of_vector_sol}
\widetilde{\Psi}_\pm = \big(\sqrt{r} K_{\overline \Lambda}[\chi_0^\pm A_\nu^\pm])\big|_{(r,\varphi)\to(x_1,x_2)}
\end{equation}
is the leading-order term of the asymptotic solution.
\end{theorem}
\begin{remark}
Conditions for~$(\omega_1,\omega_2)$ we impose can be weakened; it suffices to require $\omega_j$ to be nonresonant instead of Diophantine, i.e. $k_1 \omega_1 + k_2\omega_2 \ne 0$ for all $k = (k_1,k_2) \in\mathbb Z^2\setminus\{0\}$ such that k-th harmonic (k-th Fourier coefficient) of $\operatorname{Re}\mathcal H_{1, TW}^\pm\big|_{\overline\Lambda}$ is nonzero.
\end{remark}
\begin{remark}
$\frac{1}{2}\int W^\pm (r)\, dr$ is necessary to cancel with the imaginary part of correction in the transport equation. For~$W^+$, this factor becomes $\sqrt{U(r)-E}$.
\end{remark}

We should also note an~interesting fact that the trigonal warping correction does not affect the asymptotic spectrum of $\widehat{\mathcal L}$  due to the radial symmetry in its principal symbol. Indeed, let 
$\lambda$ be a~correction to an~asymptotic eigenvalue~$E_\nu + O(h)$ of the operator~$\widehat L$, i.e. $E_\nu +h\lambda + O(h^2)$ be an~asymptotic eigenvalue of~$\widehat{\mathcal L}$. Similarly, let $E_\nu + h\lambda_{TW}+O(h^2)$ be an~asymptotic eigenvalue of $\widehat{\mathcal L}_{TW}$.
\begin{theorem}
$\lambda = \lambda_{TW}$.
\end{theorem}
\begin{proof}
Let us look at the formulas for $\lambda$ and $\lambda_{TW}$, which follow from~\eqref{lambda_condition}:
$$
\lambda = \frac{\iint_{\overline\Lambda} \operatorname{Re} H_{1} (p,x,E,\widetilde\lambda)\big|_{\overline\Lambda,\ \widetilde\lambda=0}\, d\theta_1\, d\theta_2}{4\pi \int_0^{2\pi} (U(R(\theta_1))-E)\, d\theta_1},
\quad
\lambda_{TW} = \frac{\iint_{\overline\Lambda} \operatorname{Re} H_{1,TW} (p,x,E,\widetilde\lambda)\big|_{\overline\Lambda,\ \widetilde\lambda=0}\, d\theta_1d\theta_2}{4\pi \int_0^{2\pi} (U(R(\theta_2))-E)\, d\theta_1}.
$$
Note that $U\ne E$ everywhere, so $4\pi \int_0^{2\pi} (U(R(\theta_1))-E)\, d\theta_1 = C\ne 0$.
Now subtract one from another and substitute the formulas for $\mathcal H_1$ and $\mathcal H_{1,TW}$ in polar coordinates:
$$
\lambda_{TW}-\lambda = -2\frac{\gamma}{C}\iint_{\overline\Lambda} f(\theta_1)\cos{3\varphi(\theta_1,\theta_2)} + g(\theta_1) \sin{3\varphi(\theta_1,\theta_2)}\, d\theta_1\, d\theta_2,
$$
where $\varphi(\theta_1,\theta_2) = \theta_2 + \frac{2p_\varphi}{\omega_1} Q(\theta_1)$, and $f(\theta_1)$, $g(\theta_1)$ denote certain functions of $\theta_1$, which appear in the formula for $H_{1,TW}$. Let us perform a~variable substitution $\theta\to(t,\varphi)$ in the integral: $\theta_1 = \omega_1 t$, $\theta_2 = \varphi - \frac{2 p_\varphi}{\omega_1}Q(\theta_1) = \varphi - \frac{2 p_\varphi}{\omega_1}Q(\omega_1 t)$, so $\big|\operatorname{det}\frac{\partial \theta}{\partial (t, \varphi)}\big| = \omega_1$; thus,
$$
\lambda_{TW}-\lambda = -2\frac{\gamma}{C}\iint_{(0, 0)}^{(T,2\pi)} \big( f(\omega_1 t)\cos{3\varphi} + g(\omega_1 t) \sin{3\varphi}\big) \omega_1\, dt\, d\varphi.
$$
Integrating with respect to~$\varphi$, we finally get $\lambda_{TW}-\lambda = 0$.
\end{proof}

In practical use of the formula~\eqref{main_repr}, it would be convenient to take only a~few Fourier coefficients of the amplitude. In this case, what is the contribution of remainder of the Fourier series to the solution? In fact, this contribution is small. The following theorem provides a~mathematical justification for disregarding the tail of the series.
\begin{theorem}
Let conditions of Theorem~\ref{TheoremConstr} hold.
Set 
$$
A_N^\pm(\theta) = \frac{1}{2}\int W^{\pm}(r)\, dr\big|_{r = R(\theta)} \exp{\Big(\sum_{|k_1|\le N} \big( \frac{c_{k_1,-3}}{\omega_1 k_1 - 3\omega_2} e^{-3i\theta_2} + \frac{c_{k_1,0}}{\omega_1 k_1} + \frac{c_{k_1,3}}{\omega_1 k_1 + 3\omega_2} e^{3i\theta_2}\big)e^{i k_1 \theta_1} + \frac{i}{h} \langle q,\theta\rangle\Big)},
$$
where $c_k \equiv c_{k_1,k_2}$ are the Fourier coefficients of $\operatorname{Re} H_{1,TW}^\pm$. Then the following estimate
$$
\| \widehat{\mathcal L}\Psi_{\pm} -\mathscr E \Psi_{\pm}\|_{L^2} = O\Big(h^{3/2} + \frac{1}{N^m}\Big)
$$
holds for any $m>0$.
\end{theorem}
\begin{proof}
The correction can be written in the form
$$
\operatorname{Re} H_{1,TW}^{\pm} = f_{-3}(\theta_1) e^{-3i \theta_2} + f_{0}(\theta_1)+ f_{3}(\theta_1) e^{3i \theta_2}, \quad f_j \in C^\infty(\mathbb S^1, \mathbb C).
$$
Let us write the formula for $A_N^\pm(\theta)$ with vanishing action defect:
$$
A_N^\pm(\theta) = \frac{1}{2} \int W^\pm (r)\, dr \exp{\Big(\sum_{|k_1|\le N} \frac{c_k}{\langle \omega,\rangle k} e^{i\langle k, \theta \rangle}\Big)}.
$$
Obviously, $A_N^\pm \rightrightarrows A_\infty^\pm \equiv A^\pm$ from Theorem~\ref{TheoremConstr}.
Since functions~$f_j$ are infinitely differentiable, $c_k = O(k^{-m-4})$ for any~$m>-2$. Let us apply $(-ih \partial/\partial\theta_1)^3$ to 
$$
A_\infty^\pm - A_N^\pm = (\frac{1}{2}\int W^\pm (r)\, dr)\times \exp{\Big(\sum_{|k_1|\le N} \frac{c_k}{\langle \omega, k\rangle} e^{i \langle k,\theta\rangle} \Big)} \times \Big( \exp{\Big(\sum_{|k_1| > N} \frac{c_k}{\langle \omega, k\rangle} e^{i \langle k,\theta\rangle} \Big)} - 1\Big)
$$
and estimate it in $L^2$ norm. 
Application of the operator to the first factor yields the terms $O\Big(\frac{1}{N^{m+2}}\Big)$; 
when applied to the second and third factors, the worst-order terms take the form
$$
C \frac{k^3 c_k}{\langle \omega, k\rangle} = O(k^{-m-2}),\quad \text{which gives } O(N^{-m}) \text{ for the whole series}.
$$
Thus, we get $$ \|(-ih \partial/\partial\theta_1)^3 (A^\pm-A_N^\pm)\|_{L^2} \le C_1 O\Big(\frac{1}{N^{m+2}}\Big) + C_2 O\Big(\frac{1}{N^{m}}\Big) + C_3 O\Big(\frac{1}{N^{m}}\Big) = O\Big(\frac{1}{N^{m}}\Big).$$
\end{proof}

\section{Example}
Let us consider an~example that both could be meaningful for physics and effectively demonstrates the effect caused by trigonal warping.

Let $E_0 = 0.85$ eV, $l = 55$ nm, $\mathbf B = 7$ T, $l_B = 10$ nm, $Z = 29$ protons per core (i.e., impurity is given by some isotope of copper), $U(r) = - \frac{A}{r}$, where
$$
A = \frac{1}{E_0} \frac{1}{4\pi \epsilon_0} \frac{Z e}{l \cdot 10^{-9}} \approx 0.893663, 
$$
$$
h = \frac{6.242\cdot 10^{18} \hbar v_F}{E_0 l \cdot 10^{-9}} \approx 0.0858163,
$$
$B = \frac{v_F \mathbf{B} \hbar l 10^{-9}}{E_0} \approx 0.439353$, $\nu_1 = 16$, $\nu_2 = 15$, $\gamma = \frac{E_0}{6 t h} \approx 0.550271$, $M = 0.7$, $E = 0.935$. Then $E_\nu\approx 0.939054$ and $\lambda\approx 0.244074$, i.e., $\mathscr E = E + h\lambda \approx 0.95999952$.
Now that all the parameters are set, we can plot~$\Psi$ and~$\Psi_{TW}$ using formula~\eqref{main_repr} with $A_N$, $N=40$ in the unperturbed case and $N=70$ for the perturbed one.
The plots of first components of pseudospinors, i.e., $\operatorname{Re}\Psi_1^+$ and $\operatorname{Re}\Psi_1^{+,TW}$, are depicted in Figs.~\ref{fig3}~and~\ref{fig4} respectively.

\begin{figure}[h]
\begin{minipage}[h]{0.4\linewidth}
\center{\includegraphics[width=\linewidth]{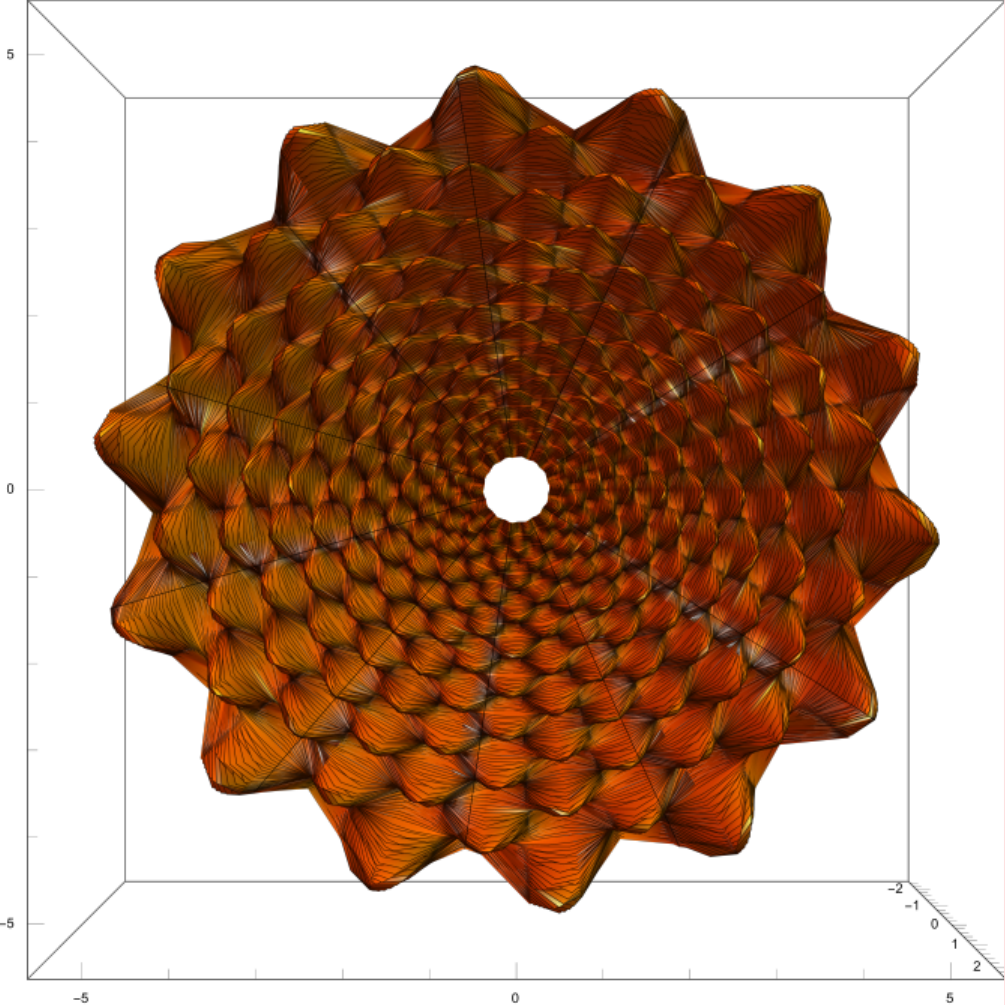}}
\caption{$\operatorname{Re}\Psi_1^+$}
\label{fig3}
\end{minipage}
\hfill
\begin{minipage}[h]{0.4\linewidth}
\center{\includegraphics[width=\linewidth]{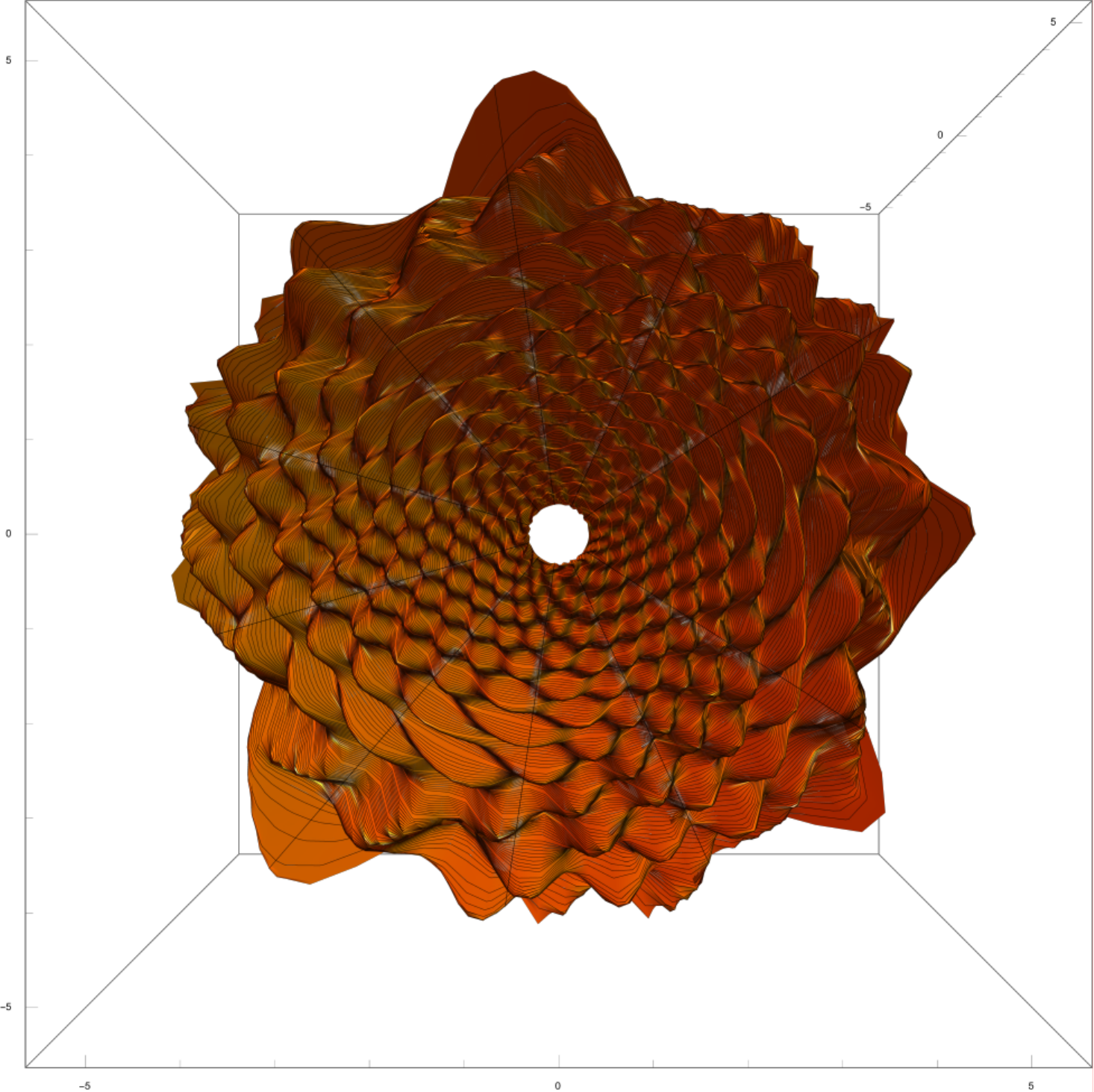}}
\caption{$\operatorname{Re}\Psi_1^{+,TW}$}
\label{fig4}
\end{minipage}
\end{figure}

The density plots of~$|\Psi|^2 = |\Psi_1|^2+|\Psi_2|^2$ depicted in Fig.~\ref{fig5}~and~\ref{fig6} (nonperturbed and perturbed cases respectively) are more illustrative in this case. These plots are expected to coincide with what is observed with STM for the corresponding bound states.
\begin{figure}[h]
\begin{minipage}[h]{0.4\linewidth}
\center{\includegraphics[width=\linewidth]{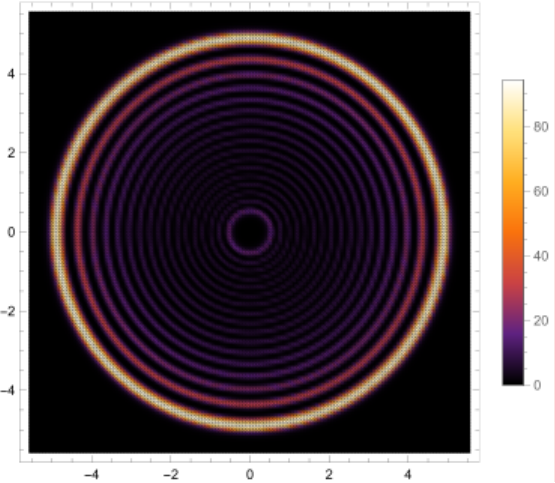}}
\caption{Density plot of $|\Psi^+|^2$}
\label{fig5}
\end{minipage}
\hfill
\begin{minipage}[h]{0.4\linewidth}
\center{\includegraphics[width=\linewidth]{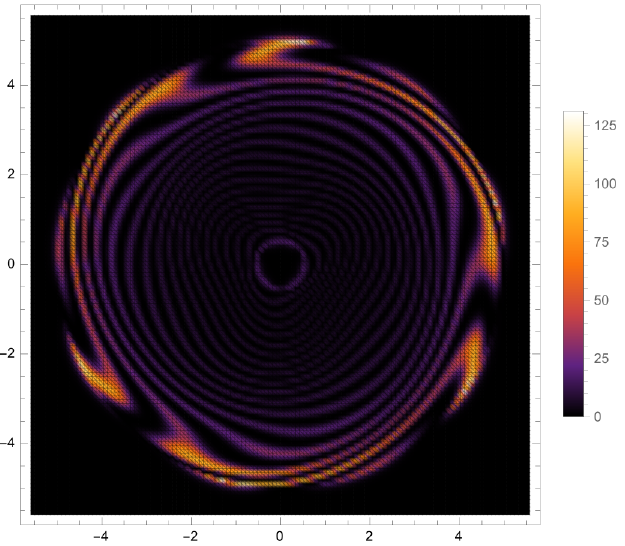}}
\caption{Density plot of $|\Psi_{TW}^+|^2$}
\label{fig6}
\end{minipage}
\end{figure}

\section*{Acknowledgments}
The author wishes to express gratitude to Dr.~Koen Reijnders from Radboud University and Dr.~Anatoly Anikin for helpful discussions and for their valuable advice.

\end{document}